\documentclass[11pt]{article}

\usepackage{epsf,amsfonts,amsmath}
\usepackage{amssymb}
\usepackage{amsthm}
\usepackage{epic,eepic}
\usepackage{texdraw}
\newtheorem{thm}{Theorem}[subsection]
\newtheorem{proposition}[thm]{Proposition}
\newtheorem{lemma}[thm]{Lemma}
\newtheorem{corollary}[thm]{Corollary}
\newtheorem{remark}[thm]{Remark}
\usepackage{texdraw}
\usepackage{mathdots}
\usepackage[dvips]{color}
\usepackage{color}
\usepackage{graphicx}
\usepackage{exscale,relsize}
\newcommand{\field}[1]{\mathbb{#1}}
\newcommand{\tr}{\mbox{tr}}
\newcommand{\diag}{\mbox{diag}}
\DeclareMathOperator{\rank}{rank}
\DeclareMathOperator{\End}{End}
\title{\textbf{On the integrability of the transfer dynamics of non-involutive Yang-Baxter maps}}
\author{S. Konstantinou-Rizos\\
\small \textit{Department of Applied Mathematics, University of Leeds, Leeds}\\
\small \texttt{mmskr@leeds.ac.uk}
}

\begin{document}

\maketitle\begin{abstract}
It is well known that, given a Yang-Baxter map, there is a hierarchy of commuting transfer maps, which arise out of the consideration of initial value problems. In this paper, we show that one can construct invariants of the transfer maps corresponding to the $n$-periodic initial value problem on the two-dimensional lattice, using the same generating function that is used to produce invariants of the Yang-Baxter map itself. Moreover, we discuss the Liouville integrability of these transfer maps. Finally, we consider four-dimensional Yang-Baxter maps corresponding to the nonlinear Schr\"odinger (NLS) equation and the derivative nonlinear Schr\"odinger (DNLS) equation which have recently appeared. We show that the associated transfer maps are completely integrable. 
\end{abstract}

\section{Introduction}
The quantum Yang-Baxter equation originates in the works of Yang \cite{Yang} and Baxter \cite{Baxter}. The study of the set-theoretical solutions of  the quantum Yang-Baxter (YB) equation was first proposed by Drinfel'd in 1992 \cite{Drinfel'd} and has been of great interest by many researchers. However, the first examples of such solutions appeared in \cite{Sklyanin}. Veselov in \cite{Veselov} proposed the more ellegant term ``Yang-Baxter maps'' for such solutions and, moreover, he connected them with integrable mappings \cite{Veselov, Veselov3}.

Yang-Baxter maps are closely related to several concepts of integrability as, for instance, the multidimensionally consinstent equations \cite{ABS-2004, ABS-2005, Bobenko-Suris,Frank3,Frank5,Frank4}. They are also related to integrable partial differential equations via Darboux transformations \cite{Sokor-Sasha}.

In this paper, we are interested in the transfer dynamics of Yang-Baxter maps which admit Lax representation \cite{Veselov2}. Therefore, we are dealing with non-involutive YB maps, as they are more interesting in terms of dynamics. Although, involutive maps have also useful applications \cite{Pavlos}.

The paper is organised as follows: Section 2 deals with \textit{parametric Yang-Baxter maps} and their \textit{Lax representations}, and we also give the definition of the Liouville integrability of a YB map, in order to make this paper self-contained. In section 3, we give a brief introduction to the transfer dynamics of Yang-Baxter maps and we show how can one derive the invariants of the \textit{transfer maps} corresponding to the $n$-\textit{periodic initial value problem} on the two dimensional lattice, using the same generating function of the invariants of the Yang-Baxter map. Furthermore, we discuss the Liouville integrability of the transfer dynamics for the $2n$-periodic problem. Finally, in section 4, we study the transfer dynamics of the Yang-Baxter maps related to the nonlinear Schr\"odinger equation and the derivative nonlinear Schr\"odinger equation.

\section{Preliminaries}
Let $A$ be an algebraic variety and $Y^{ij}$, $i,j=1,2,3$, $i\neq j$, be maps defined on the triple Cartesian product, $Y^{ij}\in \End (A\times A\times A)$. Particularly, we define $Y^{ij}$'s by the following relations
\begin{subequations}\label{Yij}
\begin{align}
Y^{12}(x,y,z)&=(u(x,y),v(x,y),z), \\ 
Y^{13}(x,y,z)&=(u(x,z),y,v(x,z)), \\ 
Y^{23}(x,y,z)&=(x,u(y,z),v(y,z)),
\end{align}
\end{subequations}
were $x,y,z\in A$. The map $Y^{ji}$, $i<j$, is defined as $Y^{ij}$ where we swap $u(k,l)\leftrightarrow v(l,k)$, $k,l=x,y,z$. For example, $Y^{21}(x,y,z)=(v(y,x),u(y,x),z)$.

A map $Y\in \End(A\times A)$ is called a \textit{Yang-Baxter map} if it satisfies the following equation
\begin{equation}\label{YB_eq1}
Y^{12}\circ Y^{13} \circ Y^{23}=Y^{23}\circ Y^{13} \circ Y^{12},
\end{equation}
which is the so-called \textit{Yang-Baxter equation}. Moreover, the map $Y$ is called \textit{reversible} if the composition of $Y^{ij}$ and $Y^{ji}$ is the identity map, namely
\begin{equation}\label{reversible}
Y^{ij}\circ Y^{ji}=Id.
\end{equation}

Furthermore, we use the term \textit{parametric YB map} when $u$ and $v$ are attached with parameters $a,b\in \field{C}$, namely $u=u(x,y;a,b)$
and $v=v(x,y;a,b)$, meaning that the following map
\begin{equation}
Y_{a,b}:(x,y;a,b)\mapsto (u,v;a,b)\equiv(u(x,y;a,b),v(x,y;a,b)),
\end{equation}
satisfies the \textit{parametric YB equation}
\begin{equation}\label{YB_eq}
Y^{12}_{a,b}\circ Y^{13}_{a,c} \circ Y^{23}_{b,c}=Y^{23}_{b,c}\circ Y^{13}_{a,c} \circ Y^{12}_{a,b}.
\end{equation}

\subsection{Lax representations and Liouville integrability of Yang-Baxter maps}
Following Suris and Veselov in \cite{Veselov2}, we call a \textit{Lax matrix for a parametric YB map} a matrix $L=L(x;c,\lambda)$ depending on a variable $x$, a parameter $c$ and a \textit{spectral parameter} $\lambda$, such that the following matrix refactorisation problem
\begin{equation} \label{eqLax}
L(u;a,\lambda)L(v;b,\lambda)=L(y;b,\lambda)L(x;a,\lambda), \quad \text{for any $\lambda \in \field{C}$,}
\end{equation}
is satisfied due to the YB map. The above is also called the \textit{Lax-equation}.

Since the Lax-equation (\ref{eqLax}) does not always have a unique solution, Kouloukas and Papageorgiou in \cite{Kouloukas2} proposed the term \textit{strong Lax matrix} for a YB map. This is when the Lax-equation is equivalent to a map
\begin{equation}\label{unique-sol}
  (u,v)=Y_{a,b}(x,y).
\end{equation}
Actually, the uniqueness of refactorisation (\ref{eqLax}) is a sufficient condition for the solutions of the Lax-equation to define a reversible YB map \cite{Veselov3} of the form (\ref{unique-sol}). In the opposite case, one may need to check if the obtained map satisfies the YB equation.

One of the most well known examples of parametric YB maps is Adler's map \cite{Adler}, which is given by
\begin{equation}\label{Adler_map}
(x,y)\longrightarrow (u,v)=\left(y+\frac{a-b}{x+y},x-\frac{a-b}{x+y}\right),
\end{equation}
and it is related to the 3-D consistent discrete potential KDV equation \cite{Frank, PNC}. This map has the following strong Lax representation \cite{Veselov2, Veselov3}
\begin{equation}
L(u;a,\lambda)L(v;b,\lambda)=L(y;b,\lambda)L(x;a,\lambda), \quad \text{for any $\lambda \in K$,}
\end{equation}
where
\begin{equation}
L(x;a,\lambda)=
\left(\begin{matrix}
x & 1 \\
x^2-a-\lambda & x
\end{matrix}\right).
\end{equation}

Now, following \cite{Fordy, Veselov4} we define integrability for YB maps.

\newtheorem{CompleteIntegrability}{Definition}[subsection]
\begin{CompleteIntegrability}
A $2N$-dimensional Yang-Baxter map,  
\begin{equation}
Y:(x_1,...,x_{2N})\mapsto (u_1,...,u_{2N}), \quad u_i=u_i(x_1,...,x_{2N}), \quad i=1,...,2N, \nonumber
\end{equation}
is said to be completely integrable or Liouville integrable if
\begin{enumerate}
%\begin{itemize}
	\item there is a Poisson matrix $J_{ij}=\left\{x_i,x_j\right\}$, of rank $2r$, which is invariant under $Y$, namely $J\circ Y=\tilde{J}$, where $\tilde{J_{ij}}=\left\{u_i,u_j\right\}$,
	\item map $Y$ has $r$ functionally independent invariants, $I_i$, namely $I_i\circ Y=I_i$, which are in involution with respect to the corresponding Poisson bracket, i.e. $\left\{I_i,I_j\right\}=0$, $i,j=1,\ldots,r$, $i\neq j$,
	\item there are $k=2N-2r$ Casimir functions, namely functions $C_i$, $i=1,\ldots,k$, such that $\left\{C_i,f\right\}=0$, for any arbitrary function $f=f(x_1,...,x_{2N})$. These are invariant under $Y$, namely $C_i\circ Y=C_i$.
%\end{itemize}
\end{enumerate}
\end{CompleteIntegrability} 

In this paper we are interested in the Liouville integrability of the transfer dynamics of non-involutive Yang-Baxter maps which admit Lax representation.

\section{Liouville integrability of transfer maps}
Veselov in \cite{Veselov,Veselov3} showed that for those Yang-Baxter maps which admit Lax representation, there are corresponding hierarchies of commuting \textit{transfer} maps which preserve the spectrum of their monodromy matrix.

In this paper, we consider the case of the transfer maps which arise out of the consideration of the initial value problem on the two-dimensional lattice, with periodic boundary conditions. Specifically, given a non-involutive parametric Yang-Baxter map, $Y_{a,b}$, with Lax matrix $L=L(x;a,b)$, it is convenient to consider an initial value problem on the two-dimensional lattice as in Fig. \ref{fig-ivp}. Particularly, we consider the initial values $x_i$ and $y_i$, $i=1,\ldots,n$, placed on the vertices of a two-dimensional lattice, with periodic boundary conditions $x_{n+1}=x_1$ and $y_{n+1}=y_1$. The vertices with values $x_i$ and $y_i$, carry the parameters $a_i$ and $b_i$, respectively. 

For the $n$-periodic initial value problem on the two-dimensional lattice, we define the \textit{transfer map}
\begin{equation}\label{transfermap}
T_n:(x_1,\ldots,x_n,y_1,\ldots,y_n)\mapsto (x_1^{(1)},\ldots,x_n^{(1)},y_1^{(1)},\ldots,y_n^{(1)}),
\end{equation}
which maps the values $x_i$ and $y_i$, $i=1,\ldots,n$, to the next level of the lattice according to following relation
\begin{equation}\label{Smap}
(x_1^{(k-1)},\ldots,x_n^{(k-1)},y_1^{(k-1)},\ldots,y_n^{(k-1)})=T_n(x_1^{(k)},\ldots,x_n^{(k)},y_1^{(k)},\ldots,y_n^{(k)}),
\end{equation}
where $x_i^{(1)}=x_i$ and $y_i^{(1)}=y_i$. Note that, the map $T_1\equiv Y_{a,b}$.

\begin{figure}[ht]
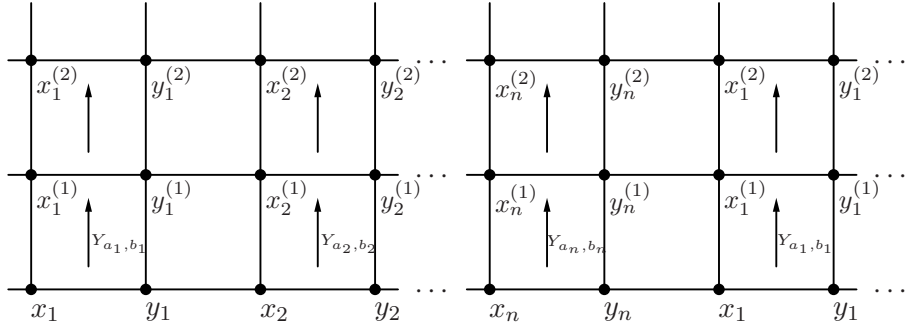

\centertexdraw{
\setunitscale 0.6
\move(-3.7 0) \lvec(-.3 0) \move(.3 0) \lvec(3.7 0)
\move(-3.7 1) \lvec(-.3 1) \move(.3 1) \lvec(3.7 1)
\move(-3.7 2) \lvec(-.3 2) \move(.3 2) \lvec(3.7 2)

\move(-3.5 0) \lvec(-3.5 2.5) \move(3.5 0) \lvec(3.5 2.5)
\move(-2.5 0) \lvec(-2.5 2.5) \move(2.5 0) \lvec(2.5 2.5)
\move(-1.5 0) \lvec(-1.5 2.5) \move(1.5 0) \lvec(1.5 2.5)
\move(-0.5 0) \lvec(-0.5 2.5) \move(0.5 0) \lvec(0.5 2.5)

\htext (-.15 -.07) {$\cdots$}
\htext (-.15 .93) {$\cdots$}
\htext (-.15 1.93) {$\cdots$}

\htext (3.85 -.07) {$\cdots$}
\htext (3.85 .93) {$\cdots$}
\htext (3.85 1.93) {$\cdots$}

\move (-3.5 0) \fcir f:0.0 r:0.05 \move (-2.5 0) \fcir f:0.0 r:0.05 \move (-1.5 0) \fcir f:0.0 r:0.05 \move (-.5 0) \fcir f:0.0 r:0.05 \move (0.5 0) \fcir f:0.0 r:0.05 \move (1.5 0) \fcir f:0.0 r:0.05 \move (2.5 0) \fcir f:0.0 r:0.05  \move (3.5 0) \fcir f:0.0 r:0.05

\move (-3.5 1) \fcir f:0.0 r:0.05 \move (-2.5 1) \fcir f:0.0 r:0.05 \move (-1.5 1) \fcir f:0.0 r:0.05 \move (-.5 1) \fcir f:0.0 r:0.05 \move (0.5 1) \fcir f:0.0 r:0.05 \move (1.5 1) \fcir f:0.0 r:0.05 \move (2.5 1) \fcir f:0.0 r:0.05  \move (3.5 1) \fcir f:0.0 r:0.05

\move (-3.5 2) \fcir f:0.0 r:0.05 \move (-2.5 2) \fcir f:0.0 r:0.05 \move (-1.5 2) \fcir f:0.0 r:0.05 \move (-.5 2) \fcir f:0.0 r:0.05 \move (0.5 2) \fcir f:0.0 r:0.05 \move (1.5 2) \fcir f:0.0 r:0.05 \move (2.5 2) \fcir f:0.0 r:0.05  \move (3.5 2) \fcir f:0.0 r:0.05

\htext (-3.5 -.28) {$x_1$} \htext (-2.5 -.28) {$y_1$} \htext (-1.5 -.28) {$x_2$} \htext (-.5 -.28) {$y_2$} \htext (.5 -.28) {$x_n$} \htext (1.5 -.28) {$y_n$} \htext (2.5 -.28) {$x_1$} \htext (3.5 -.28) {$y_1$}

\htext (-3.45 0.65) {\small$x_1^{(1)}$} \htext (-2.45 0.65) {\small$y_1^{(1)}$} \htext (-1.45 0.65) {\small$x_2^{(1)}$} \htext (-.45 0.65) {\small$y_2^{(1)}$} \htext (.55 0.65) {\small$x_n^{(1)}$} \htext (1.55 0.65) {\small$y_n^{(1)}$} \htext (2.55 0.65) {\small$x_1^{(1)}$} \htext (3.55 0.65) {\small$y_1^{(1)}$}

\htext (-3.45 1.65) {\small$x_1^{(2)}$} \htext (-2.45 1.65) {\small$y_1^{(2)}$} \htext (-1.45 1.65) {\small$x_2^{(2)}$} \htext (-.45 1.65) {\small$y_2^{(2)}$} \htext (.55 1.65) {\small$x_n^{(2)}$} \htext (1.55 1.65) {\small$y_n^{(2)}$} \htext (2.55 1.65) {\small$x_1^{(2)}$} \htext (3.55 1.65) {\small$y_1^{(2)}$}

\move(-3 0.2) \arrowheadtype t:F \arrowheadsize l:.12 w:.06  \avec(-3 0.8)  \move(-3 1.2) \arrowheadtype t:F \arrowheadsize l:.12 w:.06 \avec(-3 1.8)
\move(-1 0.2) \arrowheadtype t:F \arrowheadsize l:.12 w:.06 \avec(-1 0.8)  \move(-1 1.2) \arrowheadtype t:F \arrowheadsize l:.12 w:.06 \avec(-1 1.8)
\move(3 0.2) \arrowheadtype t:F \arrowheadsize l:.12 w:.06 \avec(3 0.8)  \move(3 1.2) \arrowheadtype t:F \arrowheadsize l:.12 w:.06 \avec(3 1.8)
\move(1 0.2) \arrowheadtype t:F \arrowheadsize l:.12 w:.06 \avec(1 0.8)  \move(1 1.2) \arrowheadtype t:F \arrowheadsize l:.12 w:.06 \avec(1 1.8)

\htext (-2.97 0.32) {\tiny$Y_{a_1,b_1}$}  
\htext (-0.97 0.32) {\tiny$Y_{a_2,b_2}$} 
\htext (3.01 0.32) {\tiny$Y_{a_1,b_1}$} 
\htext (1 0.32) {\tiny$Y_{a_n,b_n}$} 

}
\caption{{\em{$n$-periodic initial value problem}}} \label{fig-ivp}
\end{figure}

In this paper we are interested in the invariants and the Liouville integrability of the map $T_n$. 

For the Yang-Baxter map $Y_{a,b}$, the quantity $M(x,y;a,b)=L(y;b)L(x;a)$ is called the \textit{monodromy} matrix. Note that, for simplicity of the notation, we usually omit the dependence on the spectral parameter ($L(x;a,\lambda)\equiv L(x;a)$). 

The transfer map, $T_n$, has the following Lax representation 
\begin{equation}\label{LaxTn}
\prod_{i=1}^n L(u_i;a_i)L(v_i;b_i)=\prod_{i=1}^n L(y_i;b_i)L(x_i;a_i),
\end{equation}
where $L=L(x_i;a_i)$ is the Lax matrix for $Y_{a,b}$. Moreover, for $T_n$, we define the monodromy matrix
\begin{equation}\label{MonMatrix}
M_n(x,y;a,b)=\prod_{j=1}^nL(y_j;b_j)L(x_j;a_j), \qquad M_1\equiv M,
\end{equation}
Note that, for $n=1$, the matrix $M_n$ coincides with the monodromy matrix of the Yang-Baxter map.

The fact that the monodromy matrix is a generating function of first integrals is well known from the eighties; for example, see \cite{FadTah} (and the references therein). Particularly, for the invariants of the map $T_n$ as defined in (\ref{transfermap}), with Lax representation (\ref{LaxTn}), we have the following.

\begin{proposition}
The trace of the monodromy matrix, $M_n$, is a generating function of invariants for the transfer map $T_n$.
\end{proposition}
\begin{proof}
Since,
\begin{equation}\label{trL}
\tr(\prod_{i=1}^n L(u_i;a_i)L(v_i;b_i))\overset{(\ref{LaxTn})}{=}\tr(\prod_{i=1}^n L(y_i;b_i)L(x_i;a_i))=\tr(\prod_{i=1}^n L(x_i;a_i)L(y_i;b_i)),
\end{equation}
and the function $\tr(\prod_{i=1}^n L(x_i;a_i)L(y_i;b_i))$ can be written as $\sum_k \lambda^k I_k(\textbf{x},\textbf{y};\textbf{a},\textbf{b})$, where $x:=(x_1,\ldots,x_n)$, $y:=(y_1,\ldots,y_n)$, from (\ref{trL}) follows that
\begin{equation}
I_k(\textbf{u},\textbf{v};\textbf{a},\textbf{b})=I_k(\textbf{x},\textbf{y};\textbf{a},\textbf{b}),
\end{equation}
which are invariants of the map $T_n$.
\end{proof}

The invariants constitute an important tool towards the integrability of the transfer maps in the Liouville sense. However, if one attempts to extract invariants of the transfer map $T_n$ from the trace of the mondromy matrix $M_n$ with straight-forward calculation, they may not be provided with enough invariants to claim complete integrability. Moreover, the expressions we recieve from the trace of $M_n$ are usually of big lenght, when $n\geq 3$. 

In what follows, we give a proposition for deriving the invariants of the transfer map $T_n$, for arbitrary $n\in\field{N}$, by using the monodromy matrix of the Yang-Baxter map. We use it to make some conclusions regarding the Liouville integrability of map $T_n$.

\begin{proposition}\label{MonodInv}
Let $Y_{a,b}:(x,y;a,b)\mapsto (u,v;a,b)$ be a parametric YB map with Lax matrix $L=L(x;a,\lambda)$, and the trace of its corresponding monodromy matrix is given by
\begin{equation}\label{traceM1}
\tr M_1(x,y;a,b)=\lambda^m I_{m}+\ldots+\lambda I_1+I_0, \quad I_i=I_i(x,y;a,b), \quad i=0,\ldots,m.
\end{equation}
Then, the transfer map $T_n$ admits the following invariants
\begin{equation}\label{InvTnn}
I_{i,j}=I_{i}(x_j,y_j;a_j,b_j), \quad i=0,\ldots,m, \quad j=1,\ldots,n,
\end{equation}
where $x_1=x$, $y_1=y$, $a_1=a$, $b_1=b$. Moreover, the invariants $I_{i,j}$, $j=1,\ldots,n$, are functionally independent.
\end{proposition}

\begin{proof}
From (\ref{traceM1}) we have that $I_i$, $i=0,\ldots,m$, are invariants of $Y$.

The fact that the monodromy matrix, $M_n$, correspoding to the map $T_n$, can be written as
\begin{equation}
M_n(x,y)=\prod_{j=1}^nL(y_j;b_j)L(x_j;a_j)=\prod_{j=1}^n M_1(x_j,y_j;a_j,b_j),
\end{equation}
suggests that the trace of $M_1(x_j,y_j)$, $j=1,\ldots,n$, is a generating function of invariants of the map $T_n$. Indeed,
\begin{equation}
\tr M_1(x_j,y_j;a_j,b_j)=\lambda^m I_{m,j}+\ldots+\lambda I_{1,j}+I_{0,j}, \quad j=1,\ldots,n,
\end{equation}
where $I_{i,j}$, $i=0,\ldots,m$, are given by (\ref{InvTnn}). The quantities $I_{i,j}$ are invariants of the map
\begin{equation}
(x_j,y_j)\mapsto (u_j,v_j), \quad j=1,\ldots,n,
\end{equation}
and, therefore, invariants of the map
\begin{equation}
T_n:(x_1,\ldots,x_j,\ldots,x_n,y_1,\ldots,y_j,\ldots,y_n)\mapsto (u_1,\ldots,u_j,\ldots,u_n,v_1,\ldots,v_j,\ldots,v_n).
\end{equation}
The invariants $I_{i,j}$ are functionally independent, since the gradients
\begin{equation}
\nabla I_{i,j}=\frac{\partial I_i}{\partial x_i}\textbf{e}_j+\frac{\partial I_i}{\partial y_j}\textbf{e}_{j+n}, \qquad j=1,\ldots,n, 
\end{equation}
where $\textbf{e}_j$, $j=1,\ldots,n$, are orthogonal unit vectors, are linearly independent.
\end{proof}

\begin{remark} \normalfont The trace of the monodromy matrix does not guarantee that the deduced invariants are functionally independent. Therefore, for fixed $j$, the invariants $I_{i,j}$, $i=0,\ldots,m$, may not be functionally independent. However, for fixed $i$, the invariants $I_{i,j}$ are functionally independent.
\end{remark}

\begin{remark} \normalfont If the number of invariants we exctract from the trace of the monodromy matrix $M_1$ is insufficient for integrability claims, we use the monodromy matrix $M_2$, of the two-periodic problem, as generating function of invariants. Then, we construct the invariants of the $n$-periodic problem, as in Proposition \ref{MonodInv}.
\end{remark}

In what follows, we show that the Liouville integrability of the transfer map $T_2$, corresponding to a YB map of the form (\ref{unique-sol}), implies Liouville integrability of the transfer map $T_{2n}$.

\begin{lemma}\label{MatrixFormJ2}
Let $J_2$ be a $4\times 4$ Poisson matrix of rank $2$, and $C_1=x_1+y_1$, $C_2=x_2+y_2$ the corresponding Casimir functions. Then, $J_{2}$ must be of the following form
\begin{equation}\label{J2Lemma}
J_2=
\left(\begin{array}{r|c}
A & -A \\
\hline
-A & A\end{array}\right), \qquad A=\mathfrak{i}\, \sigma_2,
\end{equation}
where $\mathfrak{i}=\sqrt{-1}$ and $\sigma_2$ is the standard Pauli matrix.
\end{lemma}

\begin{proof}
If $C_1=x_1+y_1$ and $C_2=x_2+y_2$ are the Casimir functions, then the Poisson matrix, $J_2$, will satisfy the following system of equations
\begin{equation}\label{J2Cas}
\nabla C_i\cdot J_2=0, \qquad i=1,2,
\end{equation}
where the gradients of the Casimir functions are given by
\begin{equation}
\nabla C_i=\textbf{e}_i+\textbf{e}_{i+2}, \qquad i=1,2, 
\end{equation}
where $\textbf{e}_i$ are orthogonal unit vectors.

It can be readily verified that the solution of the system (\ref{J2Cas}) is given by
\begin{equation}
\left[J_2\right]_{13}=\left[J_2\right]_{24}=0 \qquad \text{and} \qquad \left[J_2\right]_{12}=\left[J_2\right]_{23}=\left[J_2\right]_{34}=-\left[J_2\right]_{14}=c,
\end{equation}
where $c$ is an arbitrary constant, which can be set equal to $-1$ without any loss of generality. Therefore, $J_2$ is given by (\ref{J2Lemma}).
\end{proof}

\begin{lemma}\label{MatrixForm}
Let $J_{2n}$ be a $4n\times 4n$ generalisation of Poisson matrix (\ref{J2Lemma}), namely 
\begin{equation}\label{J2nGeneral}
\left[J_{2n}\right]_{ij}=
\left(\begin{array}{r|c}
A_n & -A_n \\
\hline
-A_n & A_n\end{array}\right), \qquad A_n=\diag(\underbrace{\mathfrak{i}\,\sigma_2,\ldots,\mathfrak{i}\,\sigma_2}_{n}),
\end{equation}
with $\rank(J_{2n})=2n$, and $C_i=x_i+y_i$, $i=1,\ldots,2n$ are the corresponding Casimir functions. Then, the functions $f_i=f_i(x_{2i-1},x_{2i},y_{2i-1},y_{2i})$, $i=1,\ldots,n$, are in involution with respect to the Poisson matrix (\ref{J2nGeneral}).
\end{lemma}

\begin{proof}
The gradients of the functions $f_i=f_i(x_{2i-1},x_{2i},y_{2i-1},y_{2i})$ are given by
\begin{equation}
\nabla f_i=\frac{\partial f_i}{\partial x_{2i-1}}\textbf{e}_{2i-1}+\frac{\partial f_i}{\partial x_{2i}}\textbf{e}_{2i}+\frac{\partial f_i}{\partial y_{2i-1}}\textbf{e}_{2(i+n)-1}+\frac{\partial f_i}{\partial y_{2i}}\textbf{e}_{2(i+n)}.
\end{equation}
Now, multiplying $f_i$ with the Poisson matrix $J_{2n}$ we obtain
\begin{equation}
\nabla f_i\cdot J_{2n}=(-1)^i\left(\frac{\partial f_i}{\partial x_{2i}}-\frac{\partial f_i}{\partial y_{2i}}\right)(\textbf{e}_{2i-1}-\textbf{e}_{2(i+n)-1})
-(-1)^i\left(\frac{\partial f_i}{\partial x_{2i-1}}-\frac{\partial f_i}{\partial y_{2i-1}}\right)(\textbf{e}_{2i}-\textbf{e}_{2(i+n)}).
\end{equation}
Therefore, the functions $f_i$ and $f_j$, $i,j=1,\ldots,n$, $i\neq j$, are in involution with respect to $J_{2n}$, namely,
\begin{equation}
\nabla f_i\cdot J_{2n}\cdot \left(\nabla f_j\right)^T=0,
\end{equation}
since $\textbf{e}_i\cdot\textbf{e}_i=\delta_{ij}$, where $\delta_{ij}$ is the Kronecker operator.
\end{proof}

We now use the above Lemmas to discuss the Liouville integrability of the transfer maps corresponding to the $n$-periodic initial value problem. In particular, we have the following.

\begin{proposition}\label{LiouvilleIntT22}
Let $Y_{a,b}(x,y)=(u,v)$ be a parametric YB map and the corresponding transfer map $T_2$ is completely integrable map. Then, the following statements hold:
\begin{enumerate}
	\item If the corresponding Poisson matrix is of full rank, then the map $T_{2n}$, $n\in\field{N}$, is completely integrable.
  \item If the corresponding Poisson matrix has rank 2, and the Casimir functions are of the form $C_i=x_i+y_i$, $i=1,2$, then the map $T_{2n}$, $n\in\field{N}$, is completely integrable.
\end{enumerate}
\end{proposition}

\begin{proof}
\begin{enumerate}
\item Suppose that the map $T_2$, given by 
\begin{equation}\label{T22map}
T_2: (x_1,x_2,y_1,y_2)\mapsto (u_1,u_2,v_1,v_2),
\end{equation}
is completely integrable. Then suppose that $J_2$ is the Poisson matrix and that the corresponding Poisson bracket is given by
\begin{subequations}\label{poissonT22}
\begin{align}
&\left\{x_1,x_2\right\}=\alpha_{12}, \qquad\left\{y_1,y_2\right\}=\beta_{12},\qquad\left\{x_1,y_1\right\}=\gamma_{11},\\
&\left\{x_1,y_2\right\}=\gamma_{12},\qquad\left\{x_2,y_1\right\}=\gamma_{21},\qquad\left\{x_2,y_2\right\}=\gamma_{22},
\end{align}
\end{subequations}
where $\alpha_{12}=\alpha_{12}(x_1,x_2,y_1,y_2)$, $\beta_{12}=\beta_{12}(x_1,x_2,y_1,y_2)$ and $\gamma_{ij}=\gamma_{ij}(x_1,x_2,y_1,y_2)$, $i,j=1,2$.

Matrix $J_2$ has full rank, namely $\rank J_2=4$. In this case, $T_2$ must have two functionally independent invariants, say $I_i=I_i(x_1,x_2,y_1,y_2)$, $i=1,2$, which are in involution with respect to the Poisson bracket (\ref{poissonT22}), namely $\left\{I_1,I_2\right\}=0$. 

Now, according to Proposition \ref{MonodInv}, the transfer map $T_{2n}$ admits the following invariants
\begin{equation}\label{invT2n}
I_{i,k}=I_{i}(x_{2k-1},x_{2k},y_{2k-1},y_{2k}), \qquad i=1,2, \qquad k=1,\ldots,n.
\end{equation}
We consider the following Poisson bracket
\begin{subequations}\label{poissonT2n2n}
\begin{align}
&\left\{x_{2k-1},x_{2k}\right\}=\alpha_{2k-1,2k}, \qquad\left\{y_{2k-1},y_{2k}\right\}=\beta_{2k-1,2k},\qquad\left\{x_{2k-1},y_{2k-1}\right\}=\gamma_{2k-1,2k-1},\\
&\left\{x_{2k-1},y_{2k}\right\}=\gamma_{2k-1,2k},\qquad\left\{x_{2k},y_{2k-1}\right\}=\gamma_{2k,2k-1},\qquad\left\{x_{2k},y_{2k}\right\}=\gamma_{2k,2k},
\end{align}
\end{subequations}
where $\alpha_{2k-1,2k}=\alpha_{12}(x_{2k-1},x_{2k},y_{2k-1},y_{2k})$, $\beta_{2k-1,2k}=\beta_{12}(x_{2k-1},x_{2k},y_{2k-1},y_{2k-1})$ and $\gamma_{mn}=\gamma_{ij}(x_{2k-1},x_{2k},y_{2k-1},y_{2k-1})$, $i,j=1,2$, $m,n=2k-1,2k$, for $k=1,\ldots,n$. The corresponding Poisson matrix, after a permutation of the rows and columns, can be written in the form of a $4n\times 4n$ matrix which has matrix $J_2$ along its diagonal, namely
\begin{equation}\label{PoissonMatrixT22}
J_{2n}=\diag(\underbrace{J_2,\ldots,J_2}_{n}).
\end{equation}
Therefore, the rank of $J_{2n}$ is $\rank J_{2n}=n\cdot\rank J_2=4n$. Then, by construction of the Poisson matrix, $J_{2n}$, the $2n$ invariants, $I_{i,k}$, given in (\ref{invT2n}), are in involution with respect to the Poisson bracket (\ref{poissonT2n2n}), namely $\left\{I_{i,k},I_{j,k}\right\}=0$, $i,j=1,2$, $k=1,\ldots,n$. Therefore, in this case, the map $T_{2n}$ is completely integrable.

\item The rank of $J_2$ is $\rank J_2=2$. In this case, $T_2$ will have one invariant $I=I(x_1,x_2,y_1,y_2)$ and two Casimir functions $C_i=x_i+y_i$, $i=1,2$. Thus, according to Lemma \ref{MatrixFormJ2}, $J_2$ must be of the form (\ref{J2Lemma}).

Now, according to Proposition \ref{MonodInv}, $I_k=I(x_{2k-1},x_{2k},y_{2k-1},y_{2k})$, $k=1,\ldots,n$, will be invariants for the map $T_{2n}$. Since $C_1$ and $C_2$ are invariants of the map $T_2$, according to proposition \ref{MonodInv}, the quantities $C_i=x_i+y_i$, $i=1,\ldots,2n$, are invariants of the map $T_{2n}$. Moreover they are Casimir functions for the Poisson matrix $J_{2n}$, which must be of the form (\ref{J2nGeneral}). In addition, the invariants $I_k$, $k=1,\ldots,n$, are in involution with respect to $J_{2n}$, due to Lemma \ref{MatrixForm}.

The rank of $J_{2n}$ in this case is $\rank J_{2n}=n\cdot\rank J_2=2n$, $C_{i,k}$ are $2n$ Casimirs and $I_k$ are functionally independent invariants in involution. Therefore, $T_{2n}$ is completely integrable in this case as well.
\end{enumerate}
\end{proof}

In the case of four-dimensional YB maps of the form
\begin{equation}\label{4dYB}
(\textbf{x},\textbf{y})\overset{Y_{a,b}}{\longrightarrow }(\textbf{u},\textbf{v}),
\end{equation}
where $\textbf{x}:=(x_1,x_2)$, $\textbf{y}:=(y_1,y_2)$, $\textbf{u}:=(u_1,u_2)$ and $\textbf{v}:=(v_1,v_2)$, that we deal with in section 4, we have the following.

\begin{corollary}\label{Liouville4Dmaps}
If $Y_{a,b}$ is a four-dimensional parametric YB map of the form (\ref{4dYB}) and it is completely integrable, then:
\begin{enumerate}
	\item If the corresponding Poisson matrix is of full rank, then the map $T_n$, $n\in\field{N}$, is completely integrable.
  \item If the corresponding Poisson matrix has rank 2, and the Casimir functions are of the form $C_i=x_i+y_i$, $i=1,2$, then the map $T_n$, $n\in\field{N}$, is completely integrable.
\end{enumerate}
\end{corollary}

\begin{proof}
The proof is straight-forward, if one identifies the four-dimensional YB map $Y_{a,b}$ with $T_2$ in proposition \ref{LiouvilleIntT22}.
\end{proof}

In the next section we apply the above results to YB maps related to NLS type equations.

\section{YB maps for NLS type equations}
In \cite{Sokor-Sasha} we used Darboux matrices associated to Lax operators of NLS type to construct YB maps. These Lax operators correspond to a recent classification of \textit{automorphic} Lie algebras \cite{BuryPhD,Bury-Sasha,Lombardo,Lombardo-Sanders}. Specifically, we constructed six-dimensional YB maps together with their four-dimensional restrictions on invariant leaves. We also showed that the latter are integrable in the Liouville sense. 

Here, we consider the transfer dynamics of the four-dimensional maps produced in \cite{Sokor-Sasha}. In particular, we will discuss the cases of the YB maps associated with the NLS equation \cite{ZS}
\begin{equation}\label{NLS}
p_t=p_{xx}+4p^2q, \qquad q_t=-q_{xx}-4pq^2,
\end{equation}
and the DNLS equation \cite{Kaup-Newell},
\begin{equation}\label{DNLS}
p_t=p_{xx}+4(p^2q)_x, \qquad q_t=-q_{xx}+4(pq^2)_x.
\end{equation}

\subsection{The Nonlinear Schr\"odinger equation}
In the case of the NLS equation, one can derive the Adler-Yamilov map \cite{Adler-Yamilov, Kouloukas, PT}
\begin{eqnarray} \label{YB_NLS}
(\textbf{x},\textbf{y})\overset{Y_{a,b}}{\longrightarrow }\left(y_1-\frac{a -b}{1+x_1y_2}x_1,y_2,x_1,x_2+\frac{a -b}{1+x_1y_2}y_2\right).
\end{eqnarray}
The Adler-Yamilov map follows from the following strong Lax representation \cite{Sokor-Sasha}
\begin{equation} \label{lax_eq_NLS}
  L(\textbf{u};a,\lambda)L(\textbf{v};b,\lambda)=L(\textbf{y};b,\lambda)L(\textbf{x};a,\lambda),
\end{equation}
where $L$ is given by
\begin{equation} \label{laxNLS}
L(\textbf{x};a,\lambda)=\lambda \left(
\begin{matrix}
 1 & 0\\
 0 & 0
\end{matrix}\right)+\left(
\begin{matrix}
 a+x_1x_2 & x_1\\
 x_2 & 1
\end{matrix}\right), \quad \textbf{x}=(x_1,x_2).
\end{equation}
The matrix (\ref{laxNLS}) is a Darboux matrix related to the NLS equation (\ref{NLS}), where $x_1:=p(x,t)$ and $x_2:=q(x,t)$ (see \cite{Rog-Schief} and references therein).

The map (\ref{YB_NLS}) has the following invariants
\begin{equation}\label{invYBNLS} 
I_1(\textbf{x},\textbf{y})=x_1 x_2+y_1 y_2, \qquad I_2(\textbf{x},\textbf{y})=(a+x_1 x_2)(b+y_1 y_2)+x_1 y_2+x_2 y_1,
\end{equation}
which are derived from the trace of the monodromy matrix $M(\textbf{y};b,\lambda)M(\textbf{x};a,\lambda)$. Moreover, the map (\ref{YB_NLS}) is completely integrable, with the corresponding Poisson matrix given by 
\begin{equation}\label{PoissonMatNLS}
J_2:=\diag(\mathfrak{i}\,\sigma_2,\mathfrak{i}\,\sigma_2), \qquad \text{where $\sigma_2$ is the standard Pauli matrix}. 
\end{equation}

The transfer map corresponding to the $n$-periodic problem is given by
\begin{subequations}\label{MapTnnAY}
\begin{align} 
    &T_n: (\textbf{x}_1, \ldots, \textbf{x}_{n},\textbf{y}_1, \ldots, \textbf{y}_{n}) \mapsto  (\textbf{u}_1, \ldots, \textbf{u}_{n},\textbf{v}_1, \ldots, \textbf{v}_{n}),\label{MapTnnAY-1}\\
\intertext{where $\textbf{x}_i=(x_{1,i},x_{2,i})$, $\textbf{y}_i=(y_{1,i},y_{2,i})$, $\textbf{u}_i=(u_{1,i},u_{2,i})$ and $\textbf{v}_i=(v_{1,i},v_{2,i})$, $i=1,\ldots,n$, and $\textbf{u}_i$, $\textbf{v}_i$ are given by}
&u_{1,i}=y_{1,i}-\frac{a_i-b_i}{x_{1,i}+y_{2,i}}x_{1,i}, \quad u_{2,i}=y_{2,i}, \quad v_{1,i}=x_{1,i}, \quad v_{2,i}=x_{2,i}-\frac{a_i-b_i}{x_{1,i}+y_{2,i}}y_{2,i}, \quad i=1,\ldots,n, \label{MapTnnAY-2}
\end{align}
\end{subequations}
Specifically, we have the following.

\begin{proposition}\label{AYTnnInt}
\begin{enumerate}
	\item The map (\ref{MapTnnAY}) has the following invariants
	\begin{subequations}\label{AYTnninv}
	\begin{align}
&I_{1,i}(\textbf{x},\textbf{y})=x_{1,i} x_{2,i}+y_{1,i} y_{2,i}, \quad i=1,\ldots,n, \label{AYTnninv-2} \\
&I_{2,i}(\textbf{x},\textbf{y})=(a_i+x_{1,i} x_{2,i})(b_i+y_{1,i} y_{2,i})+x_{1,i} y_{2,i}+x_{2,i} y_{1,i}, \quad i=1,\ldots,n.\label{AYTnninv-1} 
  \end{align}
	\end{subequations}
	\item The map (\ref{MapTnnAY}) is completely integrable.
\end{enumerate}
\end{proposition}
\begin{proof}
\begin{enumerate}
	\item This is due to proposition (\ref{MonodInv}).
	\item We consider the following $4n\times4n$ Poisson matrix
\begin{equation}\label{PoissonTnnNLS}
J_{2n}=\diag(\underbrace{J_2,\ldots,J_2}_{n}).
\end{equation}
The rank of this matrix is $\rank(J_{2n})=n\cdot\rank(J_2)=4n$. The gradients of the invariants are given by
\begin{subequations}
\begin{align}
\nabla I_{1,i}&=x_{2,i}\textbf{e}_i+x_{1,i}\textbf{e}_{i+1}+y_{2,i}\textbf{e}_{i+n}+y_{1,i}\textbf{e}_{i+n+1}, \\
\nabla I_{2,i}&=\left[x_{2,i}(b_i+y_{i,1}y_{2,i})+y_{2,i}\right]\textbf{e}_{i}+\left[x_{1,i}(b_i+y_{i,1}y_{2,i})+y_{1,i}\right]\textbf{e}_{i+1} \\
& \left[y_{2,i}(a_i+x_{i,1}x_{2,i})+x_{2,i}\right]\textbf{e}_{i+n}+\left[y_{1,i}(a_i+x_{i,1}x_{2,i})+x_{1,i}\right]\textbf{e}_{i+n+1}.\nonumber
\end{align}
\end{subequations}
where $\textbf{e}_i$ are the orthogonal unit vectors. Obviously,
\begin{equation}
\nabla I_{1,i}\cdot J_{2n}\cdot \left(\nabla I_{2,j}\right)^t=\nabla I_{2,i}\cdot J_{2n}\cdot \left(\nabla I_{2,j}\right)^t=\nabla I_{1,i}\cdot J_{2n}\cdot \left(\nabla I_{2,j}\right)^t=0, \quad \text{for} \quad i,j=1,\ldots,n, \quad i\neq j,
\end{equation}
since $\textbf{e}_i\cdot\textbf{e}_j=\delta_{ij}$. Moreover, one can easily verify that
\begin{equation}
\nabla I_{1,i}\cdot J_{2n}\cdot \left(\nabla I_{2,i}\right)^t=0, \qquad i=1,\ldots,n.
\end{equation}
Therefore, we have $2n$ functionally independent invariants which are in involution with respect to $J_{2n}$. Thus, the map (\ref{MapTnnAY}) is completely integrable.
\end{enumerate}
\end{proof}

\subsection{The derivative Nonlinear Schr\"odinger equation}
In the case of DNLS equation, one can derive the following Yang-Baxter map (see \cite{Sokor-Sasha})
\begin{equation} \label{YB-affine}
\left(\textbf{x},\textbf{y}\right)\overset{Y_{a,b}}{\longrightarrow }\left(y_1+\frac{a-b }{a-x_1y_2}x_1, \frac{a-x_1 y_2}{b-x_1 y_2}y_2, \frac{b-x_1 y_2}{a-x_1 y_2}x_1,x_2+\frac{b-a}{b-x_1 y_2}y_2\right).
\end{equation}
This map has the strong Lax representation (\ref{lax_eq_NLS}) where $L$ is given by 
\begin{equation} \label{Darboux_Affine}
  L(\textbf{x};k;\lambda)=\lambda^{2}\left(\begin{matrix}
    k+x_1x_2 & 0\\
    0 & 0\end{matrix}\right)+\lambda\left(\begin{matrix}
    0 & x_1\\
    x_2 & 0\end{matrix}\right)+\left(\begin{matrix}
    1 & 0\\
    0 & 1\end{matrix}\right), \qquad \textbf{x}:=(x_1,x_2)
\end{equation} 
This is a Darboux matrix correspoding to the DNLS equation (\ref{DNLS}) (see \cite{SPS}).

The map (\ref{YB-affine}) has the following invariants
\begin{equation}\label{invAffine}
I_1(\textbf{x},\textbf{y})=(a+x_1x_2)(b+y_1y_2), \qquad I_2(\textbf{x},\textbf{y})=(x_1+y_1)(x_2+y_2).
\end{equation}
which are derived from the trace of the monodromy matrix $M(\textbf{y};b,\lambda)M(\textbf{x};a,\lambda)$. Moreover, the map (\ref{YB_NLS}) is completely integrable, with the corresponding Poisson matrix given by the $4\times 4$ matrix $J_2$ in Lemma \ref{MatrixFormJ2}, $\rank(J_2)=2$. The invariant $I_2$ is product of two Casimir functions $C_1=x_1+y_1$ and $C_2=x_2+y_2$.

The corresponding $n$-periodic transfer map is given by
\begin{subequations}\label{MapT2nDNLS}
\begin{align} 
    &T_n: (\textbf{x}_1, \ldots, \textbf{x}_{n},y_1, \ldots, \textbf{y}_{n}) \mapsto  (\textbf{u}_1, \ldots, \textbf{u}_{n},\textbf{v}_1, \ldots, \textbf{v}_{n}), \quad \text{where}&\label{MapT2nDNLS-1}\\
\intertext{where $\textbf{x}_i=(x_{1,i},x_{2,i})$, $\textbf{y}_i=(y_{1,i},y_{2,i})$, $\textbf{u}_i=(u_{1,i},u_{2,i})$ and $\textbf{v}_i=(v_{1,i},v_{2,i})$, $i=1,\ldots,2n$, and $\textbf{u}_i$, $\textbf{v}_i$ are given by}
&u_{1,i}=y_{1,i}+\frac{a_i-b_i }{a-x_{1,i}y_{2,i}}x_1, \qquad u_{2,i}=\frac{a_i-x_{1,i} y_{2,i}}{b_i-x_{1,i} y_{2,i}}y_{2,i},&  \\
&v_{1,i}=\frac{b_i-x_{1,i} y_{2,i}}{a_i-x_{1,i} y_{2,i}}x_{1,i}, \qquad \qquad v_{2,i}=x_{2,i}+\frac{b_i-a_i}{b_i-x_{1,i} y_{2,i}}y_{2,i}.& \label{MapT2nDNLS-2}
\end{align}
\end{subequations}

For its integrability we have the following.

\begin{proposition}\label{DNLST2nInt}
For the $4n$-dimensional map (\ref{MapT2nDNLS}) we have the following.
\begin{enumerate}
	\item  It admits the following invariants
	\begin{equation}\label{DNLST2ninv}
	I_{1,i}=I_1(\textbf{x}_i,\textbf{y}_i),\qquad C_{1,i}=\textbf{x}_{1,i}+\textbf{y}_{1,i}, \qquad \qquad C_{2,i}=\textbf{x}_{2,i}+\textbf{y}_{2,i}, \qquad i=1,\ldots,2n,
	\end{equation}
	where $I_1$ is given by (\ref{invAffine}).
	\item It is completely integrable.
\end{enumerate}
\end{proposition}
\begin{proof}
\begin{enumerate}
	\item Since the invariants of the YB map (\ref{YB-affine}) are given by $I_1$ and $C_i=x_i+y_i$, $i=1,2$, the invariants of the map $T_{n}$ are given by (\ref{DNLST2ninv}), due to proposition (\ref{MonodInv}).
	\item The map $T_{n}$ is completely integrable, with the Poisson matrix given by (\ref{J2nGeneral}), due to corollary \ref{Liouville4Dmaps}.
\end{enumerate}
\end{proof}

\section{Conclusions}
In this paper we showed how we can derive invariants for the map $T_n$, $n\in\field{N}$, using the monodromy matrix of the YB map as a generating function. Moreover, we discussed the Liouville integrability of the map $T_{2n}$. Finally, we showed that all the transfer maps associated with the Adler-Yamilov map and the YB map corresponding to the DNLS equation are completely integrable.

The results obtained in this paper can be developed in the following ways:

\begin{enumerate}
	\item study the transfer dynamics of involutive YB maps;
	\item study the transfer dynamics of the corresponding Entwining Yang-Baxter maps \cite{Kouloukas4};
	\item extend all the results for the extensions of these Yang-Baxter maps on Grassmann algebras;
\end{enumerate}

Specifically, for 1), it is clear that, considering the initial value problem on the two-dimensional lattice, as in Fig. \ref{fig-ivp}, the evolution on the lattice will be trivial for involutive YB maps. However, for involutive maps, following \cite{Kouloukas4}, one could consider $n$-periodic initial value problems, with the initial values placed on the edges of the staircase \cite{PNC} as in Fig. \ref{staircase}. This gives rise to the following transfer map
\begin{equation}
T_n^1(x_1,\ldots,x_n,y_1,\ldots,y_n)\mapsto (x_1^{(1)},\ldots,x_n^{(1)},y_2^{(1)},\ldots,y_n^{(1)},y_1^{(1)}),
\end{equation}
and the $k$-transfer map
\begin{equation}
T_n^k(x_1,\ldots,x_n,y_1,\ldots,y_n)\mapsto (x_1^{(k)},\ldots,x_n^{(k)},y_{r+1}^{(k)},\ldots,y_n^{(1)},y_1^{(k)},\ldots,y_r^{(k)}),
\end{equation}
where $r\equiv k\mod n$.

\begin{figure}[ht]
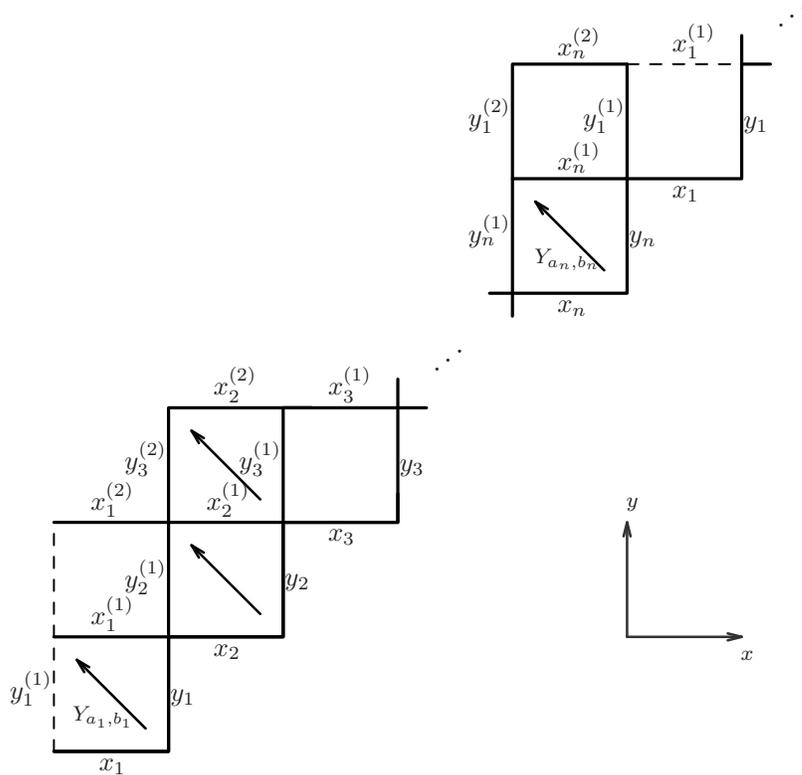

\centertexdraw{
%%%%x-y axes
\setunitscale 0.6
\move(3 -1)  \linewd 0.02 \setgray 0.2 \arrowheadtype t:V \arrowheadsize l:.12 w:.06 \avec(3 0) 
\move(3 -1) \arrowheadtype t:V  \avec(4 -1)
\setgray 0.0 
\move(-1.2 -1.8)  \arrowheadtype t:V \arrowheadsize l:.12 w:.06 \avec(-1.8 -1.2) 
\move(-0.2 -0.8)  \arrowheadtype t:V \arrowheadsize l:.12 w:.06 \avec(-0.8 -0.2) 
\move(2.8 2.2)  \arrowheadtype t:V \arrowheadsize l:.12 w:.06 \avec(2.2 2.8) 
\move(-0.2 0.2)  \arrowheadtype t:V \arrowheadsize l:.12 w:.06 \avec(-0.8 0.8) 
%\arrowheadsize l:.20 w:.10
%\move(-.5 .5) \linewd 0.02 \setgray 0.4 \arrowheadtype t:F \avec(-1.5 1.5) 
%\move(1.5 -.5) \linewd 0.02 \setgray 0.4 \arrowheadtype t:F \avec(2.5 -1.5) 
\setgray 0.0
\linewd 0.03 \move (-2 -2) \lvec (-1 -2) \lvec (-1 -1) \lvec (0 -1) \lvec (0 0) \lvec (1 0)  \lvec (1 1) \lvec(1 1.25) \move (1.8 2) \lvec (2 2) \lvec (3 2) \lvec (3 3) \lvec (4 3) \lvec (4 4)
\move (-2 -1) \lvec (-1 -1) \lvec (-1 0) \lvec (0 0)  \lvec (0 1) \lvec (1 1) \lvec (1.25 1) \move (2 1.8) \lvec (2 2) \lvec (2 3) \lvec (3 3) \lvec (3 4) 
\move (-2 -2) \lvec (-1 -2) \lvec (-1 -1) \lvec (0 -1) \lvec (0 0) \lvec (1 0) \lvec (1 0.25)
\move (2 3) \lvec (2 4) \lvec (3 4) 
\move (-2 0) \lvec (-1 0) \lvec (-1 1) \lvec (0 1) \lvec (0.25 1)
\move (4 4) \lvec (4.25 4)
\move (4 4) \lvec (4 4.25)

\linewd 0.015 \lpatt (.1 .1 ) \move (-2 -2) \lvec (-2 -1) \lvec(-2 0) \move (3 4) \lvec(4 4)
%%%Axes labels
\htext (4 -1.2) {\scriptsize{$x$}}
\htext (3 .1) {\scriptsize{$y$}}
\htext (1.3 1.3) {$\iddots$}
\htext (4.3 4.3) {$\iddots$}
\htext (-1.83 -1.79) {\scriptsize{$Y_{a_1,b_1}$}}
\htext (2.2 2.2) {\scriptsize{$Y_{a_n,b_n}$}}

%%%Initial values labels
\htext (-1.6 -2.2) {\small{$x_1$}} \htext (-0.6 -1.2) {\small{$x_2$}} \htext (0.4 -0.2) {\small{$x_3$}} \htext (2.4 1.8) {\small{$x_n$}} \htext (3.4 2.8) {\small{$x_1$}}
\htext (-0.98 -1.6) {\small{$y_1$}} \htext (0.02 -0.6) {\small{$y_2$}} \htext (1.02 0.4) {\small{$y_3$}} \htext (3.02 2.4) {\small{$y_n$}} \htext (4.02 3.4) {\small{$y_1$}}

\htext (-1.67 -0.95) {\small{$x_1^{(1)}$}} \htext (-0.66 0.05) {\small{$x_2^{(1)}$}} \htext (0.4 1.05) {\small{$x_3^{(1)}$}} \htext (2.4 3.05) {\small{$x_n^{(1)}$}} \htext (3.4 4.05) {\small{$x_1^{(1)}$}}
\htext (-2.38 -1.6) {\small{$y_1^{(1)}$}} \htext (-1.38 -0.63) {\small{$y_2^{(1)}$}} \htext (-0.38 0.4) {\small{$y_3^{(1)}$}}  \htext (1.62 2.4) {\small{$y_n^{(1)}$}} \htext (2.62 3.4) {\small{$y_1^{(1)}$}} 

\htext (-1.67 0.05) {\small{$x_1^{(2)}$}} \htext (-0.6 1.05) {\small{$x_2^{(2)}$}}  \htext (2.4 4.05) {\small{$x_n^{(2)}$}}
 \htext (-1.38 0.4) {\small{$y_3^{(2)}$}}  \htext (1.62 3.4) {\small{$y_1^{(2)}$}}
 }
\caption{{\em{$n$-periodic initial value problem}}}\label{staircase}
\end{figure}

From the way the transfer maps are defined in this case, it follows that the dynamics of involutive maps as, for instance, Adler's map (\ref{Adler_map}), will not be trivial. However, the transfer maps in this case lose the Yang-Baxter property.

For 2), one could study the dynamics of the Yang-Baxter maps which result from the matrix refactorisation problem of a Lax pair, $(M_1,M_2)$ (see \cite{Kouloukas4}), where $M_1$ and $M_2$ are Darboux matrices corresponding to different partial differential equations of NLS type. 

Regarding 3), one could extend all the results obtained in Section 3 on Grassmann algebras \cite{GSS}.

\section*{Acknowledgements}
I have the pleasure to thank T. Kouloukas, G. Grahovski and A. Mikhailov for numerous discussions and especially D. Tsoubelis for making available his computer facilities. I would also like to thank the University of Leeds for the William Wright Smith scholarship and J. Crowther for the scholarship-contribution to fees. Finally, I acknowledge six month financial support from the School of Mathematics, the University of Leeds.

\end{document}